\documentclass{article}
\usepackage{cite}
\usepackage{amsmath}
\usepackage[utf8]{inputenc}
\usepackage[english]{babel}
\usepackage{graphicx}
\usepackage{amsfonts}
\usepackage[utf8]{inputenc}
\usepackage[english]{babel}
\usepackage{bbm}
\usepackage{amsthm}
\usepackage{xcolor}
\usepackage{float}
\usepackage{url}
\usepackage{csquotes}
\usepackage{gensymb}
\usepackage{caption}
\usepackage{subcaption}
\usepackage{accents}

\usepackage{optidef}
\usepackage{epsfig}

\usepackage{hyperref}
\usepackage[ruled,vlined]{algorithm2e}
\SetArgSty{textnormal}
\def \d {{\text{d}}}

\newtheorem{theorem}{Theorem}

\newtheorem{proposition}{Proposition}
\newtheorem{definition}{Definition}

\DeclareMathOperator*{\argmin}{arg\,min}

\begin{document}
\title{Resilient UAV Traffic Congestion Control using Fluid Queuing Models}

\author{Jiazhen Zhou, Li Jin, Dengfeng Sun
\thanks{This work  has been supported by Purdue Center for Resilient Infrastructures, Systems, and Processes, and NYU Tandon School of Engineering.}
\thanks{ Jiazhen Zhou and Dengfeng Sun are with the School of Aeronautics and Astronautics, Purdue University, West-Lafayette, IN 47906 USA; Emails: \texttt{\{zhou733, dsun\}@purdue.edu}. Li Jin is with the Department of Civil and Urban Engineering, New York University Tandon School of Engineering,Brooklyn, NY, USA; Email: \texttt{lijin@nyu.edu}.}}

\maketitle

\begin{abstract}
	In this paper, we address the issue of congestion in future Unmanned Aerial Vehicle (UAVs) traffic system in uncertain weather. We treat the traffic of UAVs as fluid queues, and introduce models for traffic dynamics at three basic traffic components: single link, tandem link, and merge link. The impact of weather uncertainty is captured as fluctuation of the saturation rate of fluid queue discharge (capacity). The uncertainty is assumed to follow a continuous-time Markov process. We define the resilience of the UAV traffic system as the long-run stability of the traffic queues and the optimal throughput strategy under uncertainties. We derive the necessary and sufficient conditions for the stabilities of the traffic queues in the three basic traffic components. Both conditions can be easily verified in practiceB. The optimal throughput can be calculated via the stability conditions. Our results offer strong insight and tool for designing flows in the UAV traffic system that is resilient against weather uncertainty.
\end{abstract}

\section{Introduction}

Unmanned Aerial Vehicles (UAVs) are undergoing a vast development in recent years. By January 2018, over one million UAVs had been registered~\cite{UAV_NUM}. Commercial UAVs, for both commodity and passenger transport, are being developed and tested. Amazon has issued its patent on UAV package delivery service~\cite{kimchi}. Uber has initiated a project on urban aerial taxi~\cite{uber}. With the increasing number of UAVs applied to various tasks in urban areas, there is an urgency to integrate the UAVs into current traffic systems~\cite{balac2017towards,kopardekar2017safely}. Current transportation researches have been focused on UAVs assisted traffic~\cite{uav_assis_intelligent,uav_photo,uav_surface}, to name a few. However, UAVs themselves can form traffic in the urban airspace, and congestion can happen when the inflow exceed the capacity in the airway. UAVs, by nature, are more sensitive to weather changes. Congestion control under weather uncertainty make the UAV traffic management more challenging than other traffic.

In this paper, the problem of resilient congestion control of future UAV traffic system in urban airspace under weather uncertainty is considered. The goal of this paper is to develop a system-theoretic framework for modeling, analyzing and designing of a resilient congestion control strategy for future UAV traffic system. The main tasks of this paper are: (1) developing realistic and tractable models for UAV traffic flow, (2) analyzing the system level performance of UAV traffic in the merits of resilience. The modeling is split into three parts of work: model the UAV traffic at macro scope(system-level) for single link, integrate the single link model into traffic networks, model the weather uncertainty and its impact on the UAV traffic. For analyzing, we introduce the definition of resilience for UAV traffic, and we derive the necessary and sufficient conditions for the UAV traffic in each traffic link model to satisfy the definition.

 We adopt one class of traditional system-level model for air traffic dynamics that has been used in air transportation community, that is, the stochastic queuing model to capture the traffic dynamics in the airway. In a single queue model, each UAV is considered as a 'customer',  the airway~(link) the UAVs travel and stay in is consider as 'buffer', and finally the UAVs departure from a link at the end of the link called 'server'. The number of UAVs leaving a queue at any time instance is upper bounded by the saturation discharge rate~(capacity) of the link. The level of congestion is characterized as the queue length. In particular, we will be using the fluid queuing model. We do realize that other queuing models(e.g., M/M/1, D/E/1)~\cite{newell2013applications} are used in air traffic management~\cite{peterson1995models}, the major part of these models focus on modeling the arrival of individual customer~\cite{departurequeue,wandynamic}. Our focus is the system-level behavior of the traffic, which is the evolution of aggregated flow. Thus, fluid queuing model better fits our purpose. The stochastic differential form of the fluid queue dynamics also offers tractability in analysis.

 To extend the result from single link/queue to general queuing network, we pay our attention to modeling traffic dynamics in two additional link models: tandem link, and merge link. Tandem link, merge link, and split link are three basic link models in traffic networks. Any possible network is combinations of these three link models~\cite{van2007modeling}. For the split link, routing policy or different destinations need to be considered, which is outside the scope of the paper. In the single link or tandem link, no control action need to be taken under the default setting. In the merge link model, we considered proportional capacity allocation control policy. Based on our models, we derive necessary and sufficient conditions for stability under constant inflow. Apart from the existing results on fluid queuing models~\cite{mitra,field2007approximate}, we considered state dependent discharge rate when modeling the traffic dynamics in tandem link and merge link. This is motivated by the observation that the amount of traffic in the upstream link discharged into downstream link can be affected by the existing amount of traffic already in the links. The rule describing this observation is sometimes referred as fundamental diagram of traffic flow~\cite{mannering2007principles}.

Properly modeling the weather uncertainty and the impact of weather uncertainty is crucial to decision making in UAV traffic management. High resolution ensemble model based weather are utilized in predicting the weather for strategic planning in aviation systems~\cite{steiner2010translation}. However, such model is infeasible for our purpose for three reasons. First, we are considering UAV traffic in the urban area, where weather measurements/predictions is coarse spatially. The urban area can be considered as a single weather zone, in which the weather condition can be considered uniform. Ensemble model is simply not necessary for our case. Second, the statistical properties of the weather dynamics, such as the frequency of weather change, are lost in the ensemble model. The statistical property is critical to provide tractable system-theoretical analyze for the UAV traffic. Third, the ensemble model is simply to expensive to compute. To address these issues, continuous-time Markov chain model is used to model the switching of weather conditions. Despite the fact that the weather dynamics evolves according to complex partial differential equations, we are interested in only a few nominal weather conditions, such as sunny, windy, rainy etc,., which are meaningful for UAV traffic management system. Therefore, the continuous-time Markov chain model is further modeled to be a finite-state one.

Simple first-order Markov chain model is able to closely predict the weather pattern~\cite{wilks1999weather}. Some articles argue that the memory-less first order Markov chain model is not enough for accurate weather prediction, and proposed second-order Markov chain to achieve more accurate predictions~\cite{gates1976markov}. That being addressed, higher order finite-state Markov chain model can be transfered in to a first order one by expanding the state space~\cite{berchtold2002mixture}. Therefore, in this paper, we assume that a first order Markov chain weather model is obtained to closely predict the weather pattern, so that we can focus on the analysis for the traffic dynamics under weather uncertainty. The Markovian formulation of the weather uncertainty in air traffic study can also be found in~\cite{zhou2012performance,zhou2011stochastic}. The impact of the weather on the UAV traffic flow is modeled as the fluctuations in the capacities of traffic queues. Bad weather reduces the the maximum number of UAVs can departure from a traffic queue. Thus we let each of the weather conditions corresponds to an operational mode of the traffic queues, each operational mode defines the capacities of the queues in the traffic system. The queue length does not jump during mode switching. The resulting queuing model is called piecewise-deterministic queuing model~\cite{jin2019analysis}. Since the mode is switching according to a finite state Markov chain, our model also belongs to the class of piecewise-deterministic Markov Process\cite{davis1984piecewise}.

The term of resiliency has various meanings in control system context. The resilience considered under UAV traffic setting is defined as: (1) robust against random perturbations, (2) efficiency with robustness. The random perturbation for UAV operation will mainly come from the weather uncertainty. Based on the stochastic queuing models, the questions we try to answer is: (1) under what conditions will the traffic queue be bounded under capacity fluctuation? (2)  How can we improve the throughput of aforementioned traffic link models? To answer the first question, we first introduce the notion of stability in term of bounded moment generating function of queue length. This general notion of stability guarantees the boundedness of the first moment (expectation), which is the quantity of interest, as well as higher order moments.  Other stability notion may be used such as asymptotic stability, which implies that the queue will converge to equilibrium point with probability 1~\cite{chatterjee2007studies}.  However, we stress the fact that converging to a certain equilibrium may not have practical meaning in traffic applications. The necessary condition for each model is intuitive; it states that the inflow rate must be less or equal to the time average of the capacity of each queue. The derivation for sufficient conditions is based on well-known Foster-Lyapunov drift condition on stability analysis for Markov processes~\cite{2013stochastic,meyn1993stability}. Based on the form of the Lyapunov function we choose, we end up with sufficient conditions in bilinear inequality form which can be solved numerically. Same stability analysis considering different models in ground traffic management setting can be found in~\cite{jin2019analysis,jin2019behavior}. With the stability results we derived, we are able to answer the second question. We show that the a sequence of bilinear inequality can be solved to maximize the throughput of each model when stability is guaranteed. Feedback control or stochastic inflow can be considered under the same framework in future research.

Note that our models in differential equation form is based on the underlying fluid queuing model. This is motivated by the vast number of UAVs that will put into service in the urban airspace~\cite{kopardekar2017safely,uav_prospective}, and the dense traffic they will create~\cite{balac2017towards}. The continuous-time model can capture the fast evolution of aggregated flow in UAV traffic system. Our definition of resilience and analysis are not limited to the continuous-time model. Discrete-time models can be proposed and stability conditions can be derived under discrete version of Foster-Lyapunov drift condition~\cite{meyn1992stabilitydiscrete}. The main contribution of this paper is summarized in the following:
\begin{enumerate}
	\item We develop models for the UAV traffic dynamics under weather uncertainty in three basic components in traffic networks: single link, tandem link, and merge link, based on fluid queuing model. The weather uncertainty is modeled as continuous-time Markov process. The uncertainty impact is captured by fluctuation of capacity for each queue.

	\item We introduce the notion of resilience in terms of the stability and optimality of the queuing system under uncertainty. The stability is defined as boundedness of moment generating function of the queue length. We derive the necessary and sufficient conditions for the traffic queues to be stable in the three traffic models. The necessary condition for each model is intuitive; it states that the inflow rate must be less or equal to the time average of the capacity of each queue. The sufficient conditions are in bilinear inequality form, which can be checked by optimization softwares. Optimality is defined as maximum constant inflow rate with stability guarantee. We show that the optimal flow rate can be calculated by solving a sequence of bilinear inequality problems. Numerical examples verify our theoretical results, and illustrate the impact of uncertainty on optimal flow.

\end{enumerate}

\section{Preliminaries}\label{prelimiaries}
We first introduce mathematical notations used throughout this paper. Denote $\mathbb{R}^n_{\geq 0}$ as real vector space with dimension $n$, with each element greater or equal to $0$. Denote operator $|\cdot|$ as the cardinality if a set is considered, as 1-norm if vector space is consider unless otherwise specified. Let $\mathbb{E}(X)$ be the expectation of random variable $X$, $\exp(X)$ as the exponential function of $X$.

We introduce basic notations for fluid queuing model for UAV traffic. The most basic element in queuing system is a single queue, consisting of a buffer and a server. The buffer stores queue, and may have finite or infinite buffer size. Fluid queue receives the traffic at nonnegative inflow rate $r(i,q,t)$ [veh/hr] (vehicles per hour) and the server discharge the traffic flow at nonnegative rate $f(i,q,t)$ [veh/hr],  where $i\in \mathcal{I}$\ is the $i_{th}$ operation mode, $q$ [vel] is the existing traffic in the buffer. Let set of the modes be $\mathcal{I}$ and $|\mathcal{I}|=m$. The discharge rate is minimum of the sending flow rate $s(i,q,t)$ [veh/hr] and the mode dependent capacity $c_i$. Let uppercase $Q\in \mathcal{Q}$ and $I\in \mathcal{I}$ denote the continuous and discrete stochastic state of the stochastic process, where feasible $\mathcal{Q}$ is defined on $\mathbb{R}_{\geq 0}$. The state space of the queuing system is the Cartesian product: $\mathcal{Q}\times \mathcal{I}$. Dynamics of the queue length is then defined in an ordinary differential equation form:
\begin{align}\label{basic_queue}
	\dot{Q}(i,q,t) = r(i,q,t)-f(i,q,t),\\
	f(i,q,t) = \min\left(s\left(i,q,t\right),c_i\right).
\end{align}
This is a general description of fluid queue dynamics. This can also be used to describe the queuing system of multiple queues by considering $Q$ as a vector. 
Figure {\ref{single_queue_fig}} demonstrates single fluid queue model. In the rest of this paper, all the traffic queue system will be one or combination of the single fluid queuing models.\\
\begin{figure}[H]
	\centering
	\includegraphics[scale = 0.5]{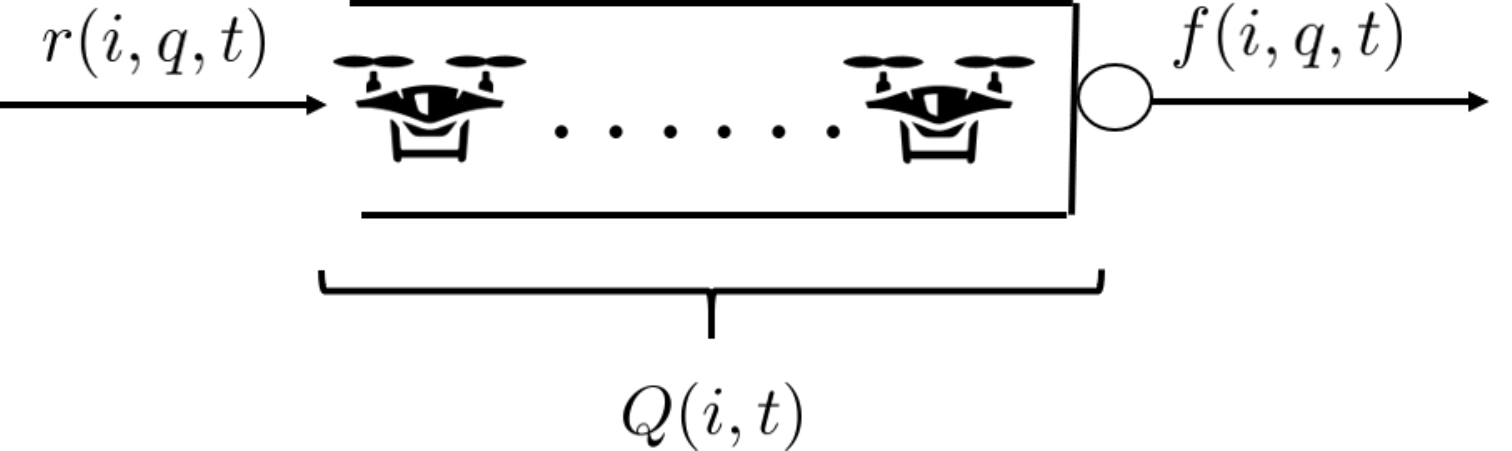}
	\caption{Single Fluid Queuing Model}
	\label{single_queue_fig}
\end{figure}
The uncertainty of system comes from the weather change, This is captured by the jump of mode $I$. Jumping between operation modes is assumed follow a continuous-time Markov process. Let the governing transition matrix be:
\begin{equation}\label{transition_matrix}
\Lambda = \begin{bmatrix}
\lambda_{11} &\lambda_{12} &\cdots &\lambda_{1m}\\
\lambda_{21} & \lambda_{22} &\cdots &\lambda_{2m}\\
\vdots &\vdots &\ddots &\vdots\\
\lambda_{m1} &\lambda_{m2} &\cdots &\lambda_{mm}
\end{bmatrix},
\end{equation}
where $\lambda_{ij},i,j\in \mathcal{I}, i\neq j$ are the time invariant jumping rate, $\lambda_{ii}=-\sum_{j\neq i}\lambda_{ij}$. We make further assumption that the process is ergodic and irreducible. This admits that the Markov process has a unique steady state distribution that satisfies the linear system:
\begin{equation}\label{transass}
p\Lambda=0,
\end{equation}
where $p=[p_1,p_2,\dots,p_m]\in \mathbb{R}^{m}_+$,and $|p|=1$~\cite{essentials}.

To define the resilient performance of the queue under uncertainty, we adopt the notion of stability in the sense of the bounded moment generating function of queue length. The formal definition is following:
\begin{definition}
	Fluid queue is said to be stable if for any initial condition $(i(0),q(0))$, there exists $C<\infty$, such that:
	\begin{align}\label{stable}
		\limsup_{t\to\infty}\frac{1}{t}\int_{ 0}^{t} \mathbb{E}[\exp(|Q(\tau)|)]d\tau\leq C.
	\end{align}
\end{definition}
This notion of stability is first introduced by Dai and Meyn~\cite{dai1995stability}. Note that expanding the exponential term using Taylor series yields boundedness for the time average of expected queue length, which shows more practical meaning.

\section{Main Results}\label{main_result}
In this section, we will derive the necessary and sufficient conditions for single queue, tandem queue, and merge queue. For latter two queuing models, we adopt the fundamental diagram of traffic to model state dependent discharge rate of queue and study the spillback effect.
\subsection{Single Queue}
Here we consider the simplest case where UAVs are congested at one point. Such case is very common in the traffic system; it can be at the ground station/warehouse where there is high volume of take off demand. In this case, we assume that the discharge of single fluid queue is simple on/off process, which is described below. We do not impose any more complicated model for single congestion because for ground station, the discharge rate is less affected by the amount of demand/inflow at ground station, UAVs can always be released at capacity at each mode. The flow diagram is:
\begin{equation}\label{flowrate}
f(i,q) = \left\{\begin{matrix}
c_i, q>0\\
\min(c_i,a), q= 0
\end{matrix}\right.
\end{equation}
where $c_i$ is the discharge rate configuration for mode $i\in \mathcal{I}$ with unit of vehicles per hour, $a$ is the constant inflow rate of UAVs.
Then the UAVs queue at single point evolves according the dynamics:
\begin{equation}\label{dynamics}
	\dot{Q}(t) = a-f(i,q)
\end{equation} 
 Under this model, there are rich theoretical literatures discussing the behavior of the queue, and we apply the results in following:
 
\begin{theorem}\label{theorem_single}
The sufficient and necessary conditions for (\ref{stable}) to hold is:
\begin{equation}\label{necessary}
	\sum_{i=1}^{m}p_ic_i\geq a.
\end{equation}
\end{theorem}
\begin{proof}
See Appendix~\ref{proof th 1}.
\end{proof}
The queue of UAVs is stable if and only if the inflow rate is less or equal to the time average of capacity for piecewise constant discharge rate, which follows the intuition.
For single fluid queue, we are able to present the steady state distribution for the queue length based on previous works~\cite{frontiers}~\cite{mitra}. Denote probability:
\small
\begin{equation*}
\begin{aligned}
	P(q,i) = &\lim_{t\to \infty} Pr\big(Q(t)\leq q, I(t)=j|Q(0)=q_0,I(0)=i\big),\\
\end{aligned}
\end{equation*}
\normalsize
and let $F(q) = [P(q,1),P(q,2),\dots,P(q,m))]$.
The spectral representation of $F(q)$ is given by:
\begin{align}
F(q) = \exp(l q)\phi,
\end{align}
where finding $l\in \mathbb{R}, \phi \in \mathbb{R}^m$ is associated with a generalized eigenproblem~\cite{frontiers}. This shows that single queue model serves as a strong tool not only in deriving conditions for bounded queue, but also giving estimations for the queue length.

\subsection{Tandem Queue with State Dependent Processing Rate}
In the air traffic system, aircrafts travel in virtual paths called links. Links define the allowable space the aircrafts can travel in. It is common that links are connected to each other. Here we pay our attention to two connected link, and study the spillback effect.
In~\cite{Jin19Trb}, Jin et al. studied the spillback effect in traffic with tandem fluid queuing model. The tandem fluid queue consists of two connecting queue where the outflow of the upstream queue is the input of the downstream queue.  The discharge rates in~\cite{Jin19Trb} is piecewise constant. In this section, we establish a new model to closer describe the nature of traffic behavior at congestion point by adopting the fundamental diagram~\cite{mannering2007principles} and study its stability conditions.

Consider a system with two queue connected in series. Traffic arrives at the upstream link with constant flow rate $a$. The outflow of upstream queue then becomes the inflow of downstream queue. The continuous stochastic state variable of the system is $Q(t) = [Q_1(t) ,\, Q_2(t)]^T$, and $q = [q_1,\, q_2]^T$ is the realization of $Q$ or existing queue length in the system. By convention, $Q\in \mathcal{Q}, \mathcal{Q}=\mathbb{R}^2_{\geq 0}$ defines the state space for continuous stochastic state. The diagram of tandem fluid queue is shown in Figure \ref{tandem_queue_fig}. 
\begin{figure}[H]
	\centering
	\includegraphics[scale = 0.32]{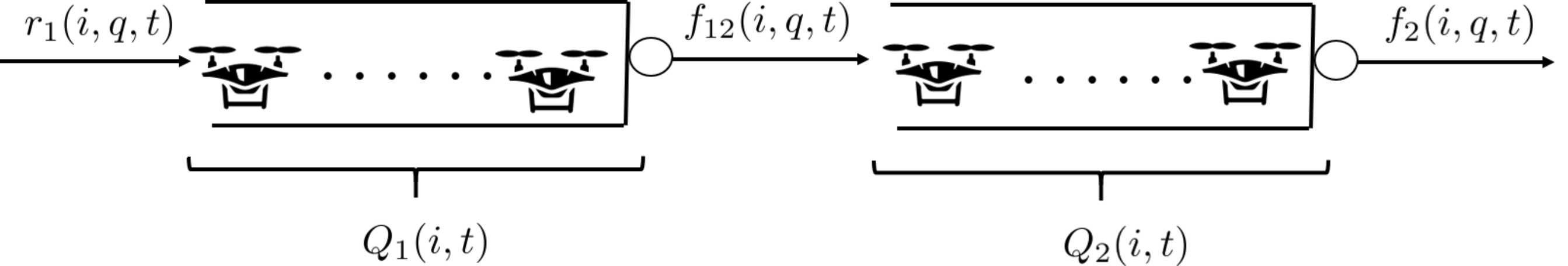}
	\caption{Tandem Fluid Queuing Model}
	\label{tandem_queue_fig}
\end{figure}
To incorporate with the fundamental traffic model, we need to change the units of the states of fluid queue model. We interpret state queue $Q$ as the traffic density of links that has the unit of vehicle per mile [veh/mile] rather than the number of waiting vehicles. Define the free flow speed of UAVs is $v$ [miles/hr], and the congested UAVs speed is $w$ [miles/hr]. Upstream queue has infinite buffer size and downstream queue has finite buffer size $\theta$ [veh/mile]. Therefore, the state takes value in the space defined by Cartesian product: $[0,+\infty)\times[0,\theta]$. The sending rates $s_1,f_2$ [veh/hr] increase with the amount of traffic in the queue, bounded by the capacity.  The flow that the downstream can receive $r_2$ decreases with the amount of traffic that is already in the queue. If downstream queue is full, then no flow can be sent from upstream to downstream. The actual flow from upstream to downstream $f_{12}$ is the minimum of $s_1$ and $r_2$.  Let $i(t)$ denote the mode at time $t$, and the associated capacity is $[c_{1i}, c_{2i}]$, where $c_{1i}$ and $c_{2i}$ is the capacity of upstream and downstream queue at mode $i$ respectively. The flow diagram is:
\begin{equation}\label{flow_diagram}
\begin{aligned}
r_1 &= a,\\
s_1(i,q) &= \min(vq_1,c_{1i}),\\
r_2(q) &= w(\theta-q_2),\\
f_{12}(i,q) & = \min(s_1,r_2),\\
f_2(i,q) &= \min(vq_2,c_{2i}).
\end{aligned}
\end{equation}
More about fundamental diagram of traffic can be found in~\cite{mannering2007principles}.   We assume that the maximum discharge rates among all modes for both queues are the same, that is:
\begin{align}
	c^{\max} = \max_{i\in \mathcal{I}} c_{ji}\quad j = 1,2.
\end{align}
And denote $c_j^{\min}=\min_{i\in \mathcal{I}}c_{ji}$.
Follow from~\cite{SSCTM}~\cite{CTM}, we also assume:
\begin{equation}
\begin{gathered}
	q_c = \frac{c^{\max}}{v},\\
	\max_{i\in \mathcal{I}}s_{1}(q_c,i)\leq r_2(q_c).
\end{gathered}
\end{equation}
along with (\ref{flow_diagram}), this gives:
\begin{align}\label{critic flow}
	c^{\max}\leq\frac{vw}{v+w}\theta.
\end{align}
Thus, by mass conservation, the evolution of system states is:
\begin{align}\label{queue_dynamics}
\dot{Q}(i,q)=F(i,q) = \begin{bmatrix} 
a-f_{12}(i,q)\\
f_{12}(i,q)-f_2(i,q)
\end{bmatrix}.
\end{align}
It can be seen that, the system is piecewise affine and continuous at each mode therefore locally Lipschitz continuous. Given the switching signal follows a irreducible and ergodic Markov process, therefore nonexplosive, unique solution for (\ref{queue_dynamics}) exists~\cite{mao1999stability}. Let the solution for system (\ref{queue_dynamics}) be $q(t)$, (\ref{queue_dynamics}) can be written as:
\begin{align}\label{time_diff}
	\dot{Q}(i,t) = \begin{bmatrix}
	a-f_{12}(t)\\f_{12}(t)-f_2(t)
	\end{bmatrix},
\end{align}
where at the derivative at jumping instance are defined as right limit. 
The queue is said to be stable if (\ref{stable}) holds.
Note that the described system is continuous, hence is also right continuous with left limit (RCLL). Therefore, the infinitesimal generator of $V(i,q)$ is given by:
\begin{align} \label{infinitesimal}
	\mathcal{L}V(i,q) = \frac{\partial V(i,q)}{\partial q}F(i,q)+\sum_{j=1}^{m}\lambda_{ij}V(i,q),
\end{align}
for any smooth function $V$ in the continuous argument~\cite{davis1984piecewise}.

In order to derive the stability conditions, we first construct the invariant set for continuous state variables, and show the sufficient condition for stability for all initial conditions in the invariant set. A set $\tilde{\mathcal{Q}}\subseteq \mathcal{Q}$ is invariant if:
\begin{align}
	\forall (q,i)  \subseteq \tilde{\mathcal{Q}}\times\mathcal{I},\, \forall t\geq 0,\, q(t)\in \tilde{\mathcal{Q}}.
\end{align}
The construction of the invariant set follows the following idea. Let the invariant set be Cartesian product $\tilde{\mathcal{Q}} = [\underline{q_1},\infty]\times[\underline{q_2},\overline{q_2}]$, such that the states governed by (\ref{dynamics}) is non-decreasing when $q_1\geq\underline{q_1},q_2\geq \underline{q_2}$ and non increasing when $q_2\leq \overline{q_2}$. 
\begin{proposition}\label{prop_tandem}
	For tandem fluid queue, with constant inflow vector $a\in R$ the set $\tilde{\mathcal{Q}} = [\underline{q_1},\infty]\times[\underline{q_2},\overline{q_2}]$ is invariant with boundaries defined as follows:
	\begin{equation}\label{q1bound}
	\begin{aligned}
		\underline{q_1} &=\min\{\frac{a}{v}, \frac{c^{\max}}{v}\} ,\\ 
		\underline{q_2} &= \min\{\underline{q_1},\frac{c_1^{\min}}{v} \},\\
		\overline{q_2} &= \theta-c_2^{\min}/w.
	\end{aligned}
	\end{equation}
\end{proposition} 
\begin{proof}
	See Appendix~\ref{Proof prop 1}.
\end{proof}
We further split the invariant set into two parts, such that $\tilde{\mathcal{Q}} = \tilde{\mathcal{Q}}_1 \cup \tilde{\mathcal{Q}}_2$, where $\tilde{\mathcal{Q}_1} = [q_c, \infty]\times[\underline{q_2},\bar{q_2}]$, $\tilde{\mathcal{Q}_2} = [\underline{q_1},q_c]\times[\underline{q_2},\bar{q_2}]$, and define
\begin{align}
	{\mathcal{F}}_1(i) = \min_{q\in \tilde{\mathcal{Q}}_1}  (f_{12}(i,q)+f_2(i,q)),\\
	{\mathcal{F}}_2(i) = \min_{q\in \tilde{\mathcal{Q}}_2}  (f_{12}(i,q)+f_2(i,q)).
\end{align}
Now we are ready to establish the necessary and sufficient conditions for the system (\ref{time_diff}) to be stable.
\begin{theorem}\label{theorem_tandem}
If system (\ref{time_diff}) is stable, then $a\leq \sum_{i = 1}^{m} c_{ji}p_i,\; j =1,2$.
	If there exist positive constants $\alpha_1,\alpha_2\dots \alpha_m$ and $\beta$ such that 
	\begin{align}\label{tandem_sufficient}
		\forall i \in \mathcal{I}, \alpha_i\beta(2a-\mathcal{F}_1(i))+\sum_{j\in \mathcal{I}}\lambda_{ij}(\alpha_j-\alpha_i)\leq -1,
	\end{align}
	then  system (\ref{time_diff}) is stable with initial condition $[q_1(0), q_2(0)]$ in $\tilde{\mathcal{Q}}$.
\end{theorem}

\begin{proof}
See Appendix~\ref{proof theorem 2}.
\end{proof}
This follows the intuition that if the upstream traffic queue does not grow to infinity, one must have the average capacity greater than the input traffic rate at each end of the queue. The sufficient condition (\ref{tandem_sufficient}) can be understood in two parts. The first term in the inequality is the effect from the dynamics of the system, and the second term is the effect from the mode transition. 
\subsection{Merge Link  with State Dependent Processing Rate}
It often shows in the traffic system that two links merge into one, or one link splits into two. Here we pay our attention to the situation where traffic comes from two sources and merge in to one link in the traffic. We will derive necessary and sufficient conditions for the merge queue. It is an extension of tandem fluid queuing model. Two upstream links are connected to one downstream link. Upstream links are modeled as infinite buffer size queue, and the downstream link is modeled as finite buffer size queue. The continuous stochastic state is $Q(t)=[Q_1(t),\,Q_2(t),\,Q_3(t)]$, where $Q_1,\,Q_2$ are lengths of the upstream queues, $Q_3$ is the length of downstream queue.  Traffic arrives two upstream queues at constant rates $[a_1,\,a_2]$ [veh/hr]. Other assumptions for tandem fluid queue in previous sections hold here. The diagram of merge link is shown in Figure~\ref{merge_queue_fig}.
 \begin{figure}[H]
 	\centering
 	\includegraphics[scale = 0.3]{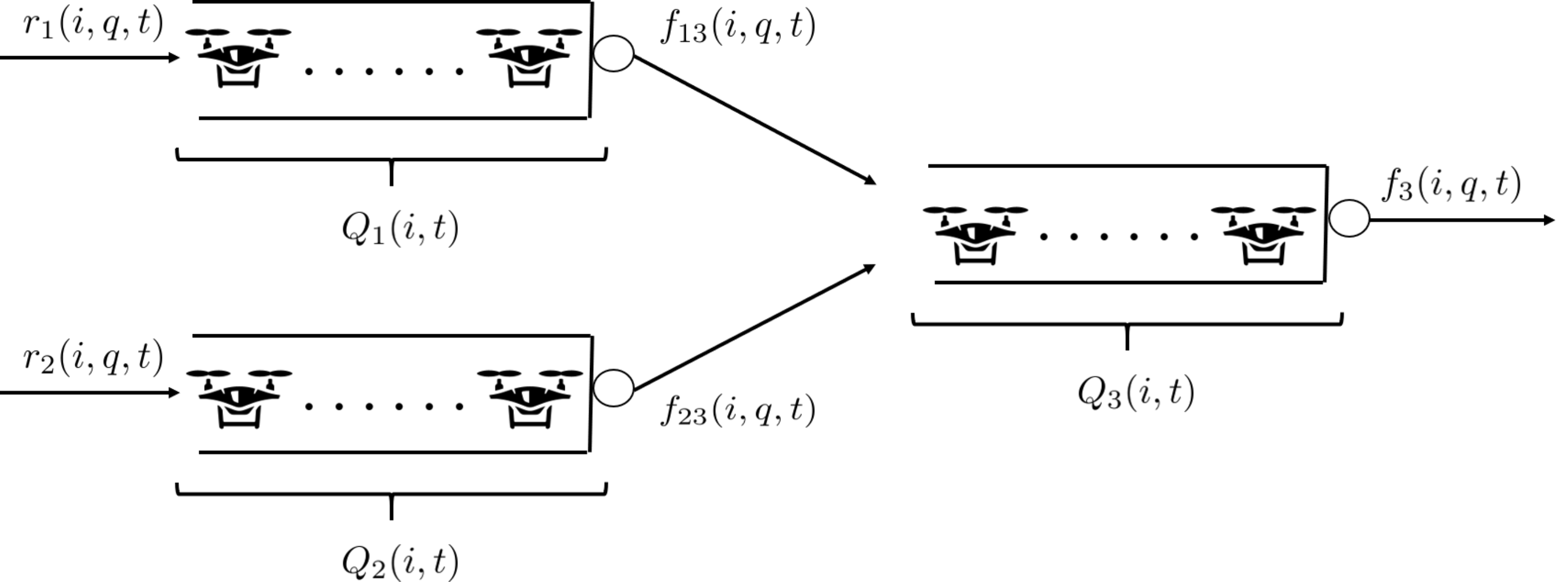}
 	\caption{Merge Fluid Queuing model}
 	\label{merge_queue_fig}
 \end{figure} 
The control action of the merging link is the inflow allocation of the downstream queue; the fraction of downstream intake allocated to a upstream discharge is proportional to the amount of traffic in the queue.
The flow diagram in short-hand thus is:
\begin{equation}
\begin{aligned}
	r_1 &= a_1,\\
	r_2 &= a_2,\\
	f_{13}(i,q) &= \min\{vq_1, \frac{q_1}{q_1+q_2}w(\theta-q_3), c_{1i}\},\\
	f_{23}(i,q) &= \min\{vq_2, \frac{q_2}{q_1+q_2}w(\theta-q_3), c_{2i}\},\\
	f_3(i,q) &= \min\{vq_3,c_{3i}\}.
\end{aligned}
\end{equation}
The evolution of states, by mass conservation, is:
\begin{align}\label{merge queue}
	\dot{Q}(i,q) = \begin{bmatrix}
	a_1-f_{13}(i,q)\\
	a_2-f_{23}(i,q)\\
	f_{13}(i,q)+f_{23}(i,q)-f_3(i,q)
	\end{bmatrix}.
\end{align}
Similarly, we can derive the stability conditions for system (\ref{merge queue}) in the sense of (\ref{stable}). We follow the same strategy for proving stability condition for tandem fluid queue. We first construct the invariant set $\tilde{\mathcal{Q}} = [\underline{q_1},\infty]\times[\underline{q_1},\infty]\times[\underline{q_3},\overline{q_3}]$, and use the monotonic property of Lyapunov function with respect to continuous argument to prove the sufficient condition for initial condition in the invariant set. The proofs for necessary and sufficient conditions are very similar. The construction for the invariant set is slightly different. We present the construction and proof for invariant set in following:
\begin{proposition}\label{Prop_merge_queue}
	For fluid queue system (\ref{merge queue}) with constant inflow vector $a_1, a_2\in R$ the set $\tilde{\mathcal{Q}} = [\underline{q_1},\infty]\times[\underline{q_2},\infty]\times[\underline{q_3},\overline{q_3}]$ is invariant with boundaries defined as follows:
	\begin{equation}\label{invar_3q}
	\begin{aligned}
		\underline{q_1} =& \min\{a_1/v, c^{\max}/v\},\\
		\underline{q_2} =&\min\{a_2/v, c^{\max}/v\},\\
		\underline{q_3} =& \min\{\underline{q_1}+\underline{q_2}, c_1^{\min}/v+ \underline{q_2},c_2^{\min}/v+ \underline{q_1},\\&c_1^{\min}/v+ c_2^{\min}/v\},\\
		\overline{q_3} =& \theta-c_3^{\min}/w.
	\end{aligned}
	\end{equation}
\end{proposition}
\begin{proof}
	See Appendix~\ref{proof prop 2}.
\end{proof}
Similarly, we let $\tilde{Q} = \tilde{Q}_1\cup\tilde{Q}_2\cup\tilde{Q}_3\cup\tilde{Q}_4$, where $\tilde{Q}_1 = [\underline{q_1},q_c]\times[\underline{q_2},q_c]\times[\underline{q_3},\overline{q_3}]$,  $\tilde{Q}_2 = [\underline{q_1},q_c]\times[q_c,\infty]\times[\underline{q_3},\overline{q_3}]$,  $\tilde{Q}_3 = [q_c,\infty]\times[\underline{q_2},q_c]\times[\underline{q_3},\overline{q_3}]$,  $\tilde{Q}_4 = [q_c,\infty]\times[q_c,\infty]\times[\underline{q_3},\overline{q_3}]$. Denote:
\begin{equation}
\begin{aligned}
		\mathcal{F}_2(i) &= \min_{q\in\tilde{Q}_2}f_{13}(i,q)+f_{23}(i,q)+f_3(i,q),\\
		\mathcal{F}_3(i) &= \min_{q\in\tilde{Q}_3}(f_{13}(i,q)+f_{23}(i,q)+f_3(i,q)),\\
		\mathcal{F}_4(i) &= \min_{q\in\tilde{Q}_4}(f_{13}(i,q)+f_{23}(i,q)+f_3(i,q)),\\
		\mathcal{F}_m(i) &= \min\{\mathcal{F}_2(i),\mathcal{F}_3(i),\mathcal{F}_4(i)\}.
\end{aligned}
\end{equation}
Solving above minimization problem can be easy. By definition:
\begin{equation}
	\begin{aligned}
	&f_{13}(i,q)+f_{23}(i,q)+f_3(i,q) =\\
	  &\min\{vq_1+vq_2,vq_1+c_{2i},vq_2+c_{1i},vq_3,c_{1i}+c_{2i},\\
	&vq_1+\frac{q_2}{q_1+q_2}w(\theta-q_3),vq_2+\frac{q_1}{q_1+q_2}w(\theta-q_3)\}\\
	&+ \min\{vq_3,c_{3i}\}
	\end{aligned}
\end{equation}
It is not hard to see that $f_{13}(i,q)+f_{23}(i,q)+f_3(i,q)$ is monotonously increasing function with respect to $q_1,\,q_2$, and concave function with respect to $q_3$. Therefore, minimum can be found by checking the boundaries of the sets.

\begin{theorem}\label{theorem_merge}
	The necessary condition for (\ref{stable}) to hold for system (\ref{merge queue})  is $a_j\leq\sum_{i=1}^{m}c_{ji}p_i$, for $j = 1, 2$, and $a_1+a_2\leq \sum_{i=1}^{m}c_{3i}p_i$. If there exist positive constants $\alpha_1,\alpha_2\dots \alpha_m$ and $\beta$ such that:
	\begin{equation}\label{merge_bili}
	\begin{aligned}
	\forall i \in \mathcal{I}, \alpha_i \beta (2a_1+2a_2-\tilde{F}_m(i))+\sum_{j\in \mathcal{I}}\lambda_{ij}(\alpha_j-\alpha_i)\\
	\leq -1
	\end{aligned}
	\end{equation}
	system (\ref{merge queue}) is stable in sense of  (\ref{stable})
\end{theorem}
The proof for necessary condition is very similar to the one for tandem queue. Define Lyapunov function as:
\begin{align}
V(i,q) = \alpha_i\exp(\beta h^Tq),
\end{align}
where $\alpha_i,\beta\in \mathbb{R}_+$, $h = [2\;2\; 1]^T$. We can prove sufficient condition for merge link with same argument in proving sufficient condition for stability in tandem queue. Therefore omitted.

The sufficient condition not only provides the method for checking if a pair of inflow rates will be stabilizing, it also offers insight for traffic flow allocation. In our model, the sufficient condition guaranteed stability for the merge queue system if the downstream intake allocation follows a proportional priority rule. Some other allocation rule may be analyzed under the same framework.

\section{Numerical Examples}\label{numerical}
In this section, we validate our theoretical findings through numerical simulations, and we show that the optimal throughput under uncertainty can be solved via solving sequence of bilinear inequality.

\subsection{Stability Conditions}
The stability conditions is verified in the merge link model. The modes are set to be one normal mode and one reduced capacity mode (congested mode),  thus, $\mathcal{I} =\{1,\,2\}$. The parameters for simulations are listed in the table below:
The parameters for numerical illustrations for tandem fluid queue is shown in Table \ref{table_merge}.
\begin{table}[H]
	\centering
	\caption{Parameters for Merge Fluid Queue.}
	\label{table_merge}
	\begin{tabular}{|c|c|c|}
		\hline
		\multicolumn{3}{|c|}{Paremeters}              \\ \hline
		Free Flow Speed         & $v$ [miles/hr]     & 8 \\ \hline
		Congested Flow Speed    & $w$ [miles/hr]     & 2 \\ \hline
		Maximum Traffic Density & $\theta$ [veh/mile] & 400 \\ \hline
		Capacity of Link 1 & $c_1$[veh/hr] & $[800,\;200]$ \\ \hline
		Capacity of Link 2 & $c_2$ [veh/hr]& $[800,\;200]$ \\ \hline
		Capacity of Link 3 & $c_3$ [veh/hr]& $[800,\;400]$ \\ \hline
	\end{tabular}
\end{table}
Let the governing transition matrix for mode switching to be:
\begin{align}
\Lambda = \begin{bmatrix}
-1& 1\\
1& -1
\end{bmatrix}.
\end{align} 
The steady state distribution for given transition matrix is $p = [0.5\;0.5]$.
To verify the necessary condition in Theorem \ref{theorem_merge}, we set the inflow to be $a_1 = 300$[veh/hr], $a_2 = 500$[veh/hr]. Therefore, $a_1+a_2> p^Tc_3 = 600$; the necessary condition is violated.  It can be seen in Figure \ref{unstable_queue} that the downstream queue $q_3$ reaches its buffer size and the upstream queues, $q_1$ and $ q_2$ grows unbounded.
\begin{figure}[H]
	\centering
	\includegraphics[scale = 0.5]{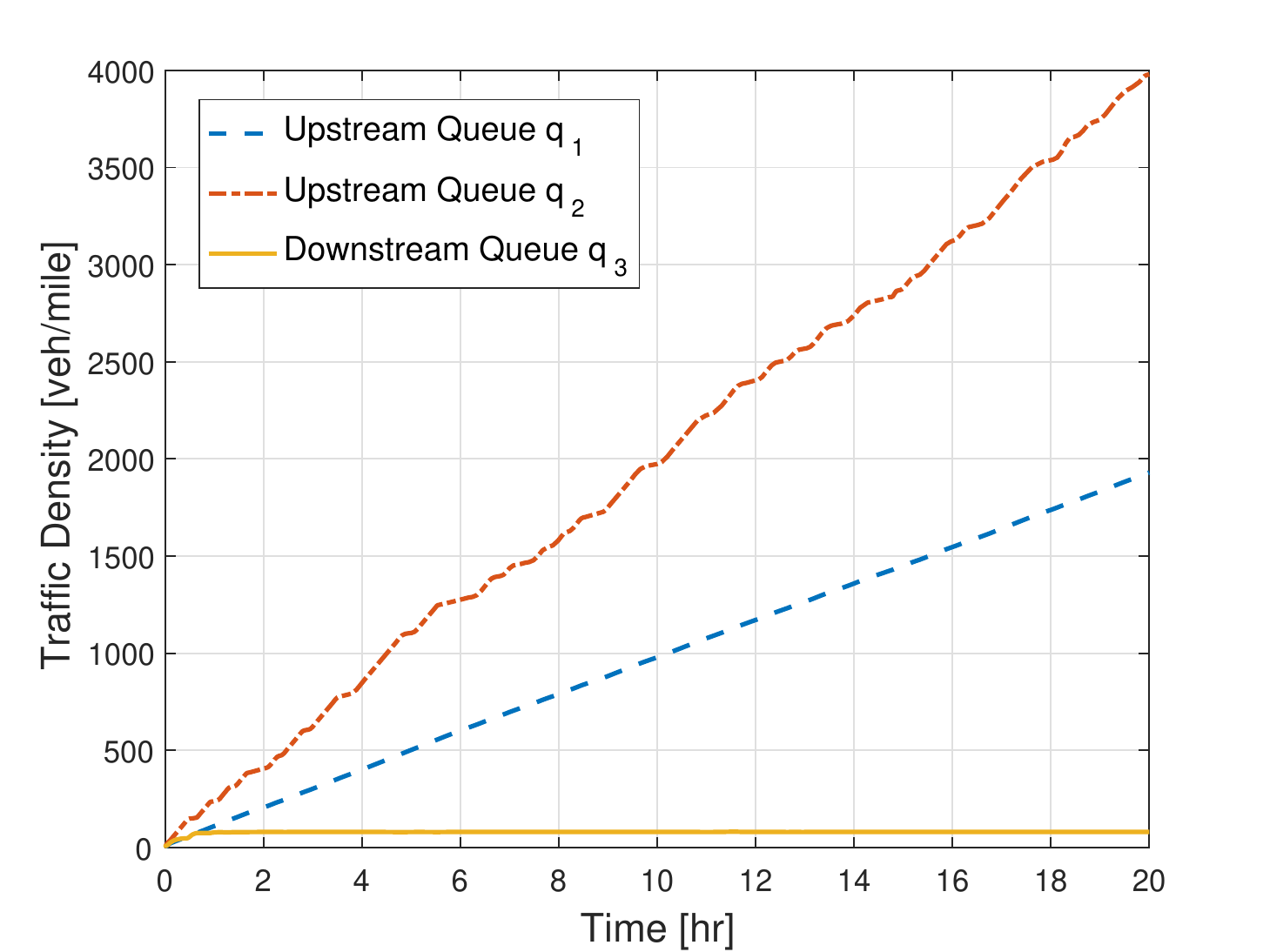}
	\caption{Unstable Queue.}
	\label{unstable_queue}
\end{figure}
To see the sufficiency in Theorem \ref{theorem_merge}, we let $a_1 = 200$[veh/hr], $a_2 = 250$[veh/hr], then the bilinear inequality is solved with YALMIP, a Matlab based optimization package~\cite{yalmip}. The solution is $\alpha_1 = 9.1956$, $\alpha_2 =12.6839$, $b = 0.1891$. The response is shown in Figure \ref{stable_queue}.
\begin{figure}[H]
	\centering
	\includegraphics[scale = 0.5]{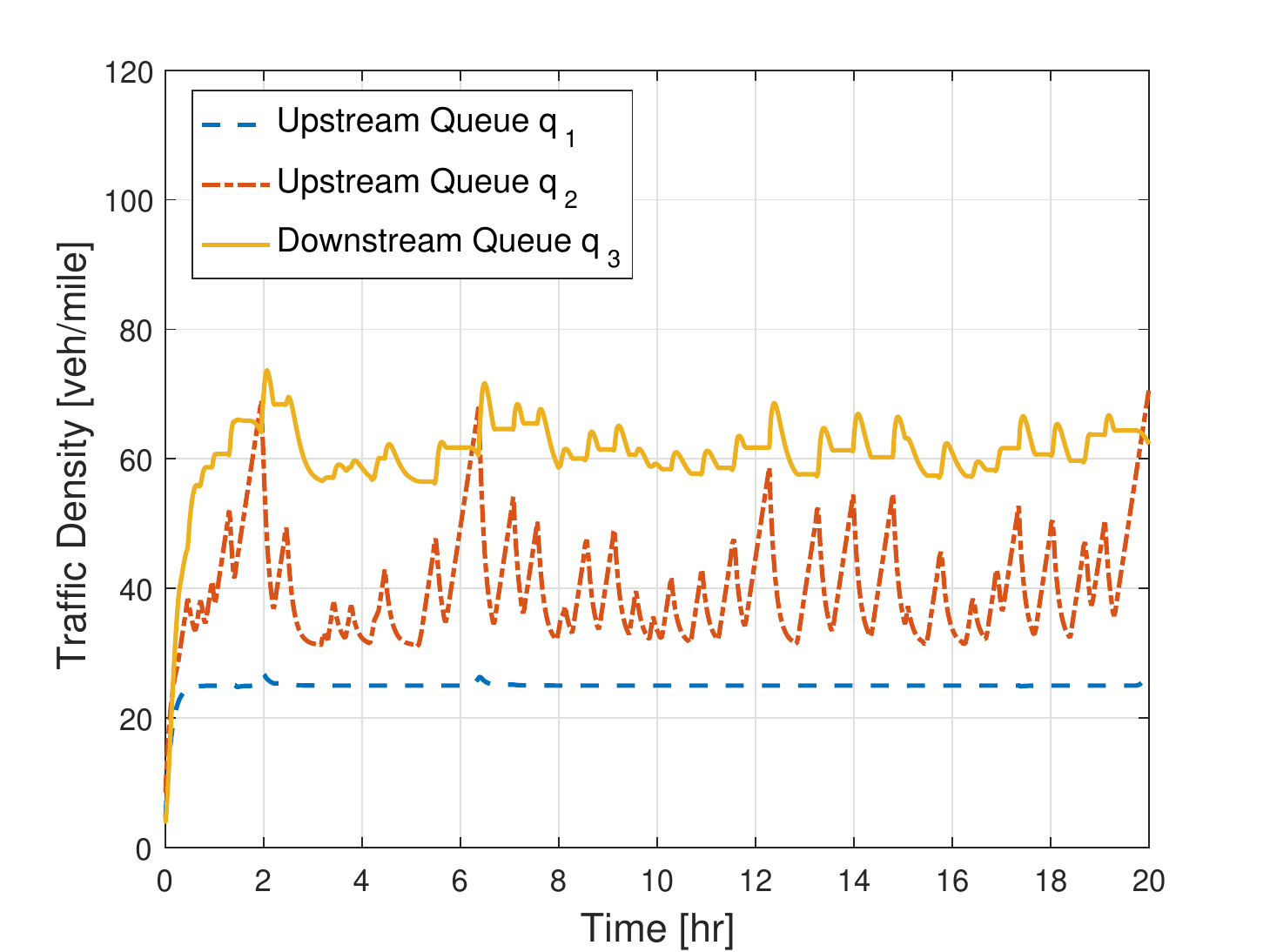}
	\caption{Stable Queue.}
	\label{stable_queue}
\end{figure}

\subsection{Optimality under Uncertainty}
In this subsection, we show how to improve the traffic throughput under weather uncertainty, and how the parameters of uncertainty will affect the throughput in the tandem fluid queuing model. Based on the form of bilinear inequality (\ref{merge_bili}), the analysis for throughput of merge link is very similar to the one in tandem link. Namely, we try to find the maximum constant inflow rate to yield stabilizing queue via sufficient condition with different parameters. This will be important for decision makers to assign traffic to each link. We claim that with the sufficient condition we derived, stability for constant inflow can be verified. We will be studying the impact of transition intensity and fluctuation on the maximum stabilizing inflow for tandem fluid queue.
 This problem can be reformulate as:
 \begin{maxi}|l|
 	{a}{a,}{}{}
 	\addConstraint{(\ref{tandem_sufficient}). }
 \end{maxi}
The maximization problem can be solved using bisection search. In each iteration of the bisection search, we solved the bilinear inequality (\ref{tandem_sufficient}).
The parameters $v,\,w,\,\theta$ for numerical illustrations for tandem fluid queue is same as the ones in Table \ref{table_merge}.
We consider two operational modes, one normal, and one congested, therefore, $\mathcal{I} =\{1,\,2\}$. The capacity of downstream link is reduced in the congested mode. Let upstream capacity be $c_1 = [800,\;800]$, downstream capacity be $c_2 = [800,\; 400]$. Elements are corresponding with normal and congested mode respectively. Let the governing transition matrix for discrete state Markov jump be:
\begin{align*}
	\Lambda = \begin{bmatrix}
	-\mu& \mu\\
	1 & -1
	\end{bmatrix}.
\end{align*}
The larger $\mu$ is, the more frequent the mode of tandem fluid queue will switch to congested mode. The steady state distribution is $p = \frac{1}{1+\mu}[1\; \mu]^T$, with which we can compute the inflow rate for necessary condition (\ref{tadem_nece}). Denote $a_{n}$ and $a_{s}$ the maximum inflow rate for necessary condition to hold and the maximum inflow rate for sufficient condition to hold respectively.

\begin{figure}[H]
	\centering
	\includegraphics[scale = 0.5]{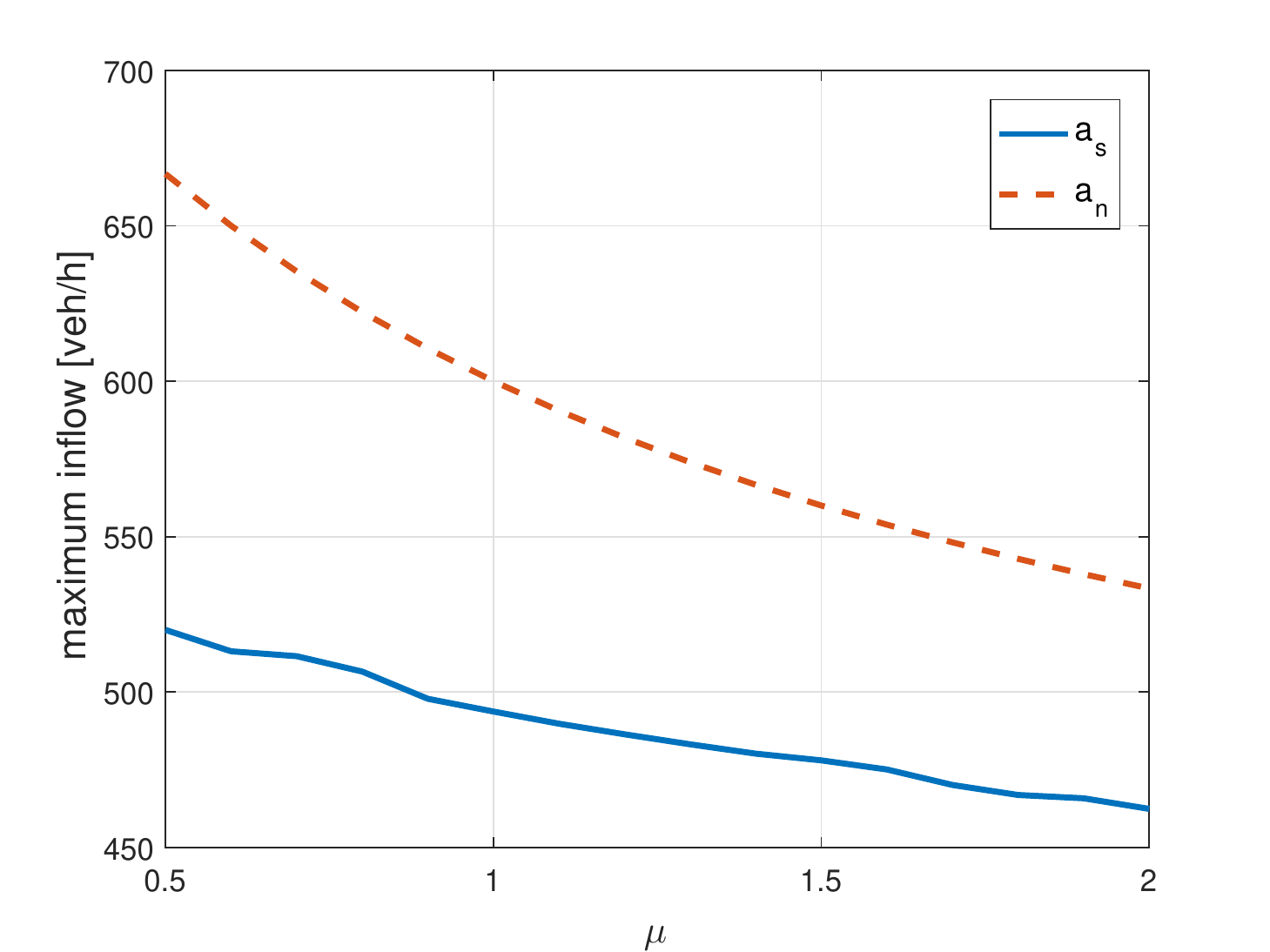}
	\caption{The change of inflow for necessary condition to hold $a_n$ and inflow for sufficient condition to hold $a_s$ with respect to transition intensity parameter $\mu$.}
	\label{max_inflow_trasition_fig}
\end{figure}
Figure \ref{max_inflow_trasition_fig} shows $a_n$ and $a_s$ with varying transitional intensity. It follows the intuition that when congested mode is visited more frequently, both inflow rate for necessary condition to hold and inflow rate for sufficient condition to hold will decrease.

Different magnitude of capacity fluctuation under same transitional intensity may yield the same necessary bound for inflow rate, however, this is not true for inflow rate for sufficient condition to hold.  We demonstrate this property by numerical example. 
Let state transition matrix be:
\begin{align}
\Lambda = \begin{bmatrix}
-1& 1\\
1& -1
\end{bmatrix}.
\end{align} 
Let the capacity fluctuation be  $c_1 = [800,\;800]$, $c_2 = [600-\delta_c,\; 600+\delta_c]$. The steady state distribution of mode is $p = [0.5 \; 0.5]$. Thus, the $a_n = 600$ for any $\delta_c$. Figure \ref{max_inflow_fluc} shows that the stabilizing inflow rate decrease with the magnitude of the fluctuation.
\begin{figure}[H]
	\centering
	\includegraphics[scale = 0.5]{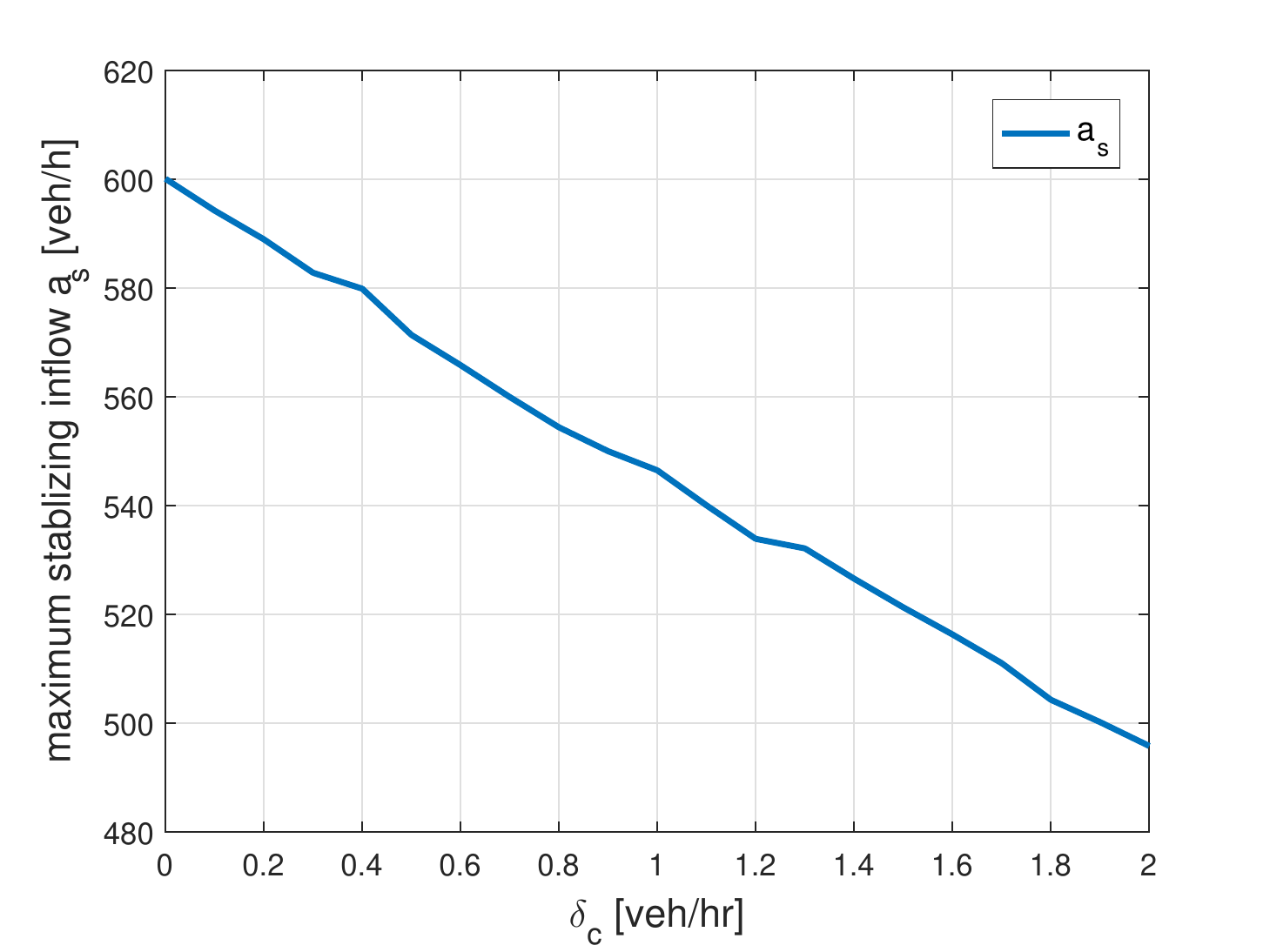}
	\caption{Stabilizing Inflow Rate $a_s$ with respect to Capacity Fluctuation Magnitude $\delta_c$.}
	\label{max_inflow_fluc}
\end{figure}

\section{Conclusion}\label{conclusion}
In this paper, we consider the congestion in future UAV traffic system under weather uncertainty from a system-theoretical perspective. We developed models for UAV traffics dynamics in three basic traffic link components. Based on our models, we derived the necessary and sufficient conditions on the inflow rates for the long run stability. We show that the necessary condition is intuitive, and the sufficient condition can be numerically verified. With our results, future UAV traffic designer can have insight and method on designing the flows in UAV traffic network and improve performance with stability guarantees.

\appendix

\subsection{Proof of Theorem \ref{theorem_single}}\label{proof th 1}
First we prove necessity. Integrating (\ref{dynamics}) gives:
\begin{align}\label{Qt}
Q(t) = \int_{0}^{t} \big(a-f(i,q)\big)\d\tau+Q(0).
\end{align}
By~\cite{stability}, it is necessary for (\ref{stable}) that the first moment follows:
\begin{align}\label{1st}
\lim_{t\to\infty}\frac{1}{t}|\mathbb{E}Q(t)-\mathbb{E}Q(0)|=0.
\end{align}
Combining (\ref{1st}) and (\ref{Qt}) gives:
\begin{equation}
\begin{aligned}
\lim_{t\to\infty}\frac{1}{t}\int_{0}^{t}\big(a-f(i,q)\big)\d\tau=0,\\
\lim_{t\to\infty}\frac{1}{t}\int_{0}^{t}a\d\tau-\lim_{t\to\infty}\frac{1}{t}\int_{0}^{t}f(i,q)\d\tau = 0.
\end{aligned}
\end{equation}
Define total duration time at mode $i$ to time $t$ as $T_i(t)=\int_{0}^{t}\mathbbm{1}_{(I(\tau)=i)}\d\tau$, where $\mathbbm{1}_{(\cdot)}$ is indicator function.
Under assumption (\ref{transass}), and by~\cite{2013stochastic}:
\begin{align}\label{tfrac}
\lim_{t\to\infty}\frac{T_i(t)}{t}=p_i.
\end{align}
By definition (\ref{flowrate}), we have:
\begin{align}\label{flowinq}
a=\lim_{t\to\infty}\sum_{i\in \mathcal{I}}\frac{T_if(i,q)}{t}\leq	\lim_{t\to\infty}\sum_{i\in \mathcal{I}}\frac{T_ic_i}{t}.
\end{align}
Combining (\ref{flowinq}) and(\ref{tfrac}) gives (\ref{necessary}). 
The proof for sufficiency is given by~\cite{frontiers}.

\subsection{Proof for Proposition \ref{prop_tandem}}\label{Proof prop 1}
	For $q_1 = \underline{q_1}$, we have:
\begin{equation}
\begin{aligned}
\dot{Q}_1(i,\underline{q_1}) &= a-\min(v\underline{q_1},c_{1i},w(\theta-q_2))\\
&\geq a-v\underline{q_1} = 0.
\end{aligned}
\end{equation}

for $q\in\tilde{\mathcal{Q}}$, $q_2 = \underline{q_2}$, we have:
\begin{equation}
\begin{aligned}
\dot{Q}_2(i,\underline{q}) &= \min(vq_1,c_{1i},w(\theta-q_2)) - \min(vq_2,c_{2i})\\
&\geq \min(v\underline{q_1},c_{1i},w(\theta-q_2)) - \min(vq_2,c_{2i}).
\end{aligned}
\end{equation}
Since $q_2\leq c_1^{\min}/v\leq c^{\max}/v$, with (\ref{critic flow}), we get:
\begin{align}
v\underline{q_2}\leq w(\theta-\underline{q_2}).
\end{align}
Therefore,
\begin{align}
\dot{Q}_2(i,\underline{q}) \geq v\underline{q_2}-v\underline{q_2}= 0.
\end{align}
Then we consider the upper boundary, $q\in\tilde{\mathcal{Q}}$, $q_2 = \overline{q_2}$. From (\ref{critic flow}):
\begin{equation}
\begin{aligned}
\frac{vw}{v+w}\theta\geq c^{\max} \geq c_2^{\min}\\
v(\theta - c_2^{\min}/w)\geq c_2^{\min}\\
\min(v\overline{q_2},c_{2i}) \geq c_2^{\min}.
\end{aligned}
\end{equation}
Along with (\ref{queue_dynamics}) gives:
\begin{equation}
\begin{aligned}
&\dot{Q}_2(i,[q_1,\overline{q_2}]^T) = \min(vq_1,c_{1i},w(\theta-\overline{q_2}))-\\
&\min(c_{2i},v\overline{q_2})\leq w(\theta-\overline{q_2}) - c_2^{\min} = 0.
\end{aligned}
\end{equation}

\subsection{Proof for Theorem \ref{theorem_tandem}}\label{proof theorem 2}

We first prove necessity for stability. Integrating (\ref{time_diff}), we get:
\begin{equation}\label{tadem_nece}
\begin{aligned}
Q_1(t) &= \int_{\tau = 0}^{t} \big(a-f_{12}(\tau)\big)\d\tau +q_1(0),\\
Q_2(t) &= \int_{\tau = 0}^{t}\big( f_{12}(\tau)-f_2(\tau)\big)\d\tau +q_2(0).
\end{aligned}
\end{equation}
Therefore,
\begin{equation}
\begin{aligned}
\lim_{t\to\infty}\frac{1}{t}\Big(|Q_1(t)|-\int_{\tau = 0}^{t} a-f_{12}(\tau)\d\tau\Big) &= \\
\lim_{t\to\infty}\frac{1}{t} Q_1(0) &= 0,\\
\lim_{t\to\infty}\frac{1}{t}\Big(|Q_2(t)|-\int_{\tau = 0}^{t} f_{12}-f_{2}(\tau)\d\tau\Big) &= \\
\lim_{t\to\infty}\frac{1}{t} Q_2(0) &= 0.
\end{aligned}
\end{equation}
(\ref{stable}) is essentially the bound for the moment generating function of $|Q(t)|$, we have $P(\lim_{t\to \infty}Q_j(t)) = 0,\; j = 1,2.$ Thus,
\begin{equation}\label{av_q}
\begin{aligned} 
\lim_{t\to\infty}\frac{1}{t}\int_{\tau = 0}^{t} \big(a-f_{12}(\tau)\big)\d\tau = 0 \quad a,s,\\
\lim_{t\to\infty}\frac{1}{t}\int_{\tau = 0}^{t} \big(f_{12}(\tau)-f_2(\tau)\big)\d\tau = 0 \quad a,s.
\end{aligned}
\end{equation}
For each mode $i$, define the total duration time at each mode to $t$:
\begin{align}
T_i(t) = \int_{\tau = 0}^{t} \mathbbm{1}_{I(\tau)=i} \d\tau.
\end{align}
By~\cite{2013stochastic}:
\begin{align}
\lim_{t\to\infty}\frac{T_i(t)}{t}=p_i,
\end{align}
Thus,
\begin{equation}
\begin{aligned}
0 &= \lim_{t\to\infty}\frac{1}{t}\int_{\tau = 0}^{t}\big( a-f_{12}(\tau)\big)\d\tau \\
&= a - \lim_{t\to\infty}\frac{1}{t}\int_{\tau = 0}^{t} \mathbbm{1}_{I(\tau)=i}f_{12}(\tau)
\\ &\geq a-\sum_{i = 1}^{m} c_{1i}p_i.
\end{aligned}
\end{equation}
From (\ref{av_q}), we can also get $0 \geq a - \sum_{i = 1}^{m} c_{2i}p_i$.

The proof for sufficient condition is derived based on a general results from Meyn and Tweedie~\cite{meyn1993stability}. The result states follow:\\
	If there exists a radially unbounded function\footnote{A function is radially unbounded if $||x||\to \infty \implies f(x)\to \infty$ for any norm on $\mathbb{R}^n$} $V: \mathcal{I}\times\tilde{\mathcal{Q}}\to\mathbb{R}_{\geq0}$, such that for some $c>0, d<\infty$
	\begin{align}\label{drift}
	\mathcal{L}V(i,q)\leq -c V(i,q)+d
	\end{align}
	then, for any initial condition $(i,q)\in \mathcal{I}\times \tilde{\mathcal{Q}}$,
	\begin{equation}\label{V_bound}
	\begin{aligned} 
	\limsup_{t\to\infty}\frac{1}{t}\int_{t}^{\tau = 0}\mathbb{E}[V(I(t),Q(t))|I(0)=i,Q(0)=q]\d\tau\\
	\leq d/c,
	\end{aligned}
	\end{equation}
where $V:\, \mathcal{I}\times \mathcal{Q}\to \mathbb{R}_+$ is called Lyapunov function, and condition (\ref{drift}) is called drift condition. By defining Lyapunov function as $V(i,q) = a_i\exp(bh^Tq) $, where $h$ is element-wise positive. It can be seen that (\ref{V_bound}) implies stability (\ref{stable}).

	Consider the Lyapunov function:
\begin{align} \label{lyapunov function}
V(i,q) = \alpha_i\exp(\beta h^Tq) 
\end{align}
where $s = [2,1]^T$.By(\ref{infinitesimal}) and (\ref{drift}) gives:
\begin{equation}\label{v+c}
\begin{aligned}
&\mathcal{L}V(i,q) = \\
&\big(\alpha_i\beta h^TF(i,q)+\sum_{j\in \mathcal{I}}\lambda_{ij}(\alpha_j-\alpha_i)\big)\frac{V(i,q)}{a_i},\\
&\mathcal{L}V(i,q) +c V(i,q) = \\
&\big(c \alpha_i+\alpha_i\beta h^TF(i,q)
+\sum_{j\in \mathcal{I}}\lambda_{ij}(\alpha_j-\alpha_i)\big)\frac{V(i,q)}{\alpha_i}.
\end{aligned}
\end{equation}
Maximizing left hand side of second equation in (\ref{v+c}) over $q\in \tilde{Q}$:
\begin{equation}
\begin{aligned}
&\max_{q\in \tilde{Q}} (\mathcal{L}V(i,q) +c V(i,q)) \\
&= \max_{q\in \tilde{Q}} \big(c \alpha_i+2\alpha_i\beta a-\alpha_i\beta f_{12}(i,q)-\alpha_i\beta f_2(i,q)\\&+\sum_{j\in \mathcal{I}}\lambda_{ij}(\alpha_j-\alpha_i)\big)\frac{V(i,q)}{\alpha_i}\\
&=\big( c \alpha_i+2\alpha_i \beta a-\alpha_i \beta\min_{q\in \tilde{Q}}(f_{12}(i,q)+f_2(i,q))\\&+\sum_{j\in \mathcal{I}}\lambda_{ij}(\alpha_j-\alpha_i)\big)\frac{V(i,q)}{\alpha_i}.
\end{aligned}
\end{equation}
We first consider for $q\in\tilde{\mathcal{Q}}_1$. Minimization on $\tilde{Q}_1$ can be easily done since the solution of this problem must lie on the boundaries of $\tilde{\mathcal{Q}}_1$. To see this, by definition (\ref{flow_diagram})
\begin{equation}
\begin{aligned}
f_{12}(i,q)+f_2(i,q) = \min\{vq_1,c_{1i},w(\theta - q_2)\} \\
+ \min\{vq_2,c_{2i}\}.
\end{aligned}
\end{equation}
This is a piecewise linear function with respect to $q_1$ with positive coefficient, therefore monotonously increasing with respected to $q_1$ with upper bound; we have $\argmin_{q_1\in[q_c,\infty]} ( f_{12}(i,q)+f_2(i,q)) = q_c$. The above function is a sum of two upper bounded concave function with respect to $q_2$,  therefore the $\argmin_{q_2\in[\underline{q_2},\bar{q_2}]} ( f_{12}(i,q)+f_2(i,q)) \in \{\underline{q_2},\bar{q_2}\}$. Thus the minimization problem can be solved by enumeration. Let $c = \frac{1}{\max \alpha_i}$, $d = 0$, then (\ref{drift}) is satisfied.\\
For $q\in \tilde{Q}_2$, we have:
\begin{equation}
\begin{aligned}
&\max_{q\in \tilde{Q}_2} (\mathcal{L}V(i,q) +c V(i,q))  \\
=& \big( c \alpha_i+2\alpha_i \beta a-\alpha_i \beta\mathcal{F}_2(i)
+\sum_{j\in \mathcal{I}}\lambda_{ij}(\alpha_j-\alpha_i)\big)\frac{V(i,q)}{\alpha_i}\\
\leq& | c \alpha_i+2\alpha_i \beta a-\alpha_i \beta \mathcal{F}_2(i)+\sum_{j\in \mathcal{I}}\lambda_{ij}(\alpha_j-\alpha_i)|\frac{V(i,q)}{\alpha_i}.
\end{aligned}
\end{equation}
Define:
\begin{equation}
\begin{aligned}
d =&\frac{\max_{q\in \tilde{Q}_2, i\in \mathcal{I}}V(i,q)}{\min_{i\in \mathcal{I}}\alpha_i}\\
 \times&\max_{i\in \mathcal{I}} | c \alpha_i+2\alpha_i \beta a-\alpha_i \beta {\mathcal{F}}_2(i)+
\sum_{j\in \mathcal{I}}\lambda_{ij}(\alpha_j-\alpha_i)|,
\end{aligned}
\end{equation}
then  (\ref{drift}) is satisfied. $V$ and ${\mathcal{F}}_2$ are bounded on $\tilde{Q}_2$, therefore maximizations are well defined.

\subsection{Proof for Proposition \ref{Prop_merge_queue}} \label{proof prop 2}
		For $q_1 = \underline{q_1}$, we have:
\begin{equation}
\begin{aligned}
\dot{Q}_1(i,\underline{q_1}) &= a-\min(v\underline{q_1},c_{1i}, \frac{q_1}{q_1+q_2}w(\theta-q_3))\\
&\geq a-v\underline{q_1} = 0.
\end{aligned}
\end{equation}
The proof for second equality is similar. The proof for third equality is slight different. By definition,
\begin{equation}
\begin{aligned}
&\dot{Q_3}(t,i) = \min\{q_1, \frac{q_1}{q_1+q_2}w(\theta-q_3), c_{1i}\}+\\
&\min\{vq_2, \frac{q_2}{q_1+q_2}w(\theta-q_3), c_{2i}\}- \min\{vq_3,c_{3i}\}
\end{aligned}
\end{equation}
By the property of minimum operator, we have:
\begin{equation}
\begin{aligned}
\dot{Q_3}(t,i)=&{} \min\{vq_1+vq_2,vq_1+c_{2i},vq_2+c_{1i},vq_3,\\
&c_{1i}+c_{2i},vq_1+\frac{q_2}{q_1+q_2}w(\theta-q_3),\\
&vq_2+\frac{q_1}{q_1+q_2}w(\theta-q_3)\}- \min\{vq_3,c_{3i}\}
\end{aligned}
\end{equation}
For $Q\in \tilde{\mathcal{Q}}$, we have:
\begin{equation}
\begin{aligned}
\dot{Q_3}(t,i) \geq&{} \min\{v\underline{q_1}+v\underline{q_2},v\underline{q_1}+c_{2i},v\underline{q_2}+c_{1i},v\underline{q_3},\\
&c_{1i}+c_{2i}, vq_1+\frac{q_2}{q_1+q_2}w(\theta-\underline{q_3}),\\
&vq_2+\frac{q_1}{q_1+q_2}w(\theta-\underline{q_3})\}- \min\{vq_3,c_{3i}\}
\end{aligned}
\end{equation}
By (\ref{critic flow}), $v\underline{q_3}\leq w(\theta -\underline{q_3})$. Along with (\ref{invar_3q}), we have:
\begin{equation}
\begin{aligned}
\underline{q_3}&\leq \underline{q_1}+\underline{q_2}\\
&\leq q_1+q_2.\\
\end{aligned}
\end{equation}
Then,
\begin{equation}
\begin{aligned}
\frac{vq_1}{q_1+q_2}\underline{q_3}&\leq vq_1,\\
(1-\frac{q_2}{q_1+q_2})v	\underline{q_3} &\leq vq_1,\\
v\underline{q_3} & \leq vq_1+v\frac{q_2}{q_1+q_2}\underline{q_3}.
\end{aligned}
\end{equation}
Similarly, $v\underline{q_3}  \leq vq_2+v\frac{q_1}{q_1+q_2}\underline{q_3}$, therefore by third equality in (\ref{invar_3q}):
\begin{align}
\dot{Q_3}(t,i)&\geq v\underline{q_3}-v\underline{q_3} \geq 0.
\end{align}	
For the upper boundary, we have:
\begin{equation}
\begin{aligned}
\dot{Q_3}(t,i) &\leq \frac{q_1}{q_1+q_2}w(\theta-\underline{q_3})+\frac{q_2}{q_1+q_2}w(\theta-\underline{q_3}) \\
& - \min\{vq_3,c_{3i}\}\leq w(\theta - \underline{q_3})-c_3^{\min}\leq 0.
\end{aligned}
\end{equation}

\bibliographystyle{unsrt}
\bibliography{reference}
\end{document}